\documentclass[a4paper,USenglish]{lipics}
\usepackage{microtype}%if unwanted, comment out or use option "draft"

\usepackage{amssymb,amsmath,amsfonts,amssymb,textcomp,mathrsfs,dsfont}
\usepackage{graphicx}
\usepackage{algorithm}
\usepackage{algorithmic}
\usepackage{psfrag}

\usepackage{color}

\newcommand{\vr}[1]{\boldsymbol{#1}}

\usepackage{tikz}
\usetikzlibrary{calc,patterns,decorations.pathmorphing,decorations.markings}

\usepackage{xspace}
\newcommand{\software}{\textsf{FAUST}$^{\mathsf 2}$\xspace}

\newtheorem{assumption}{Assumption}
\title{Dynamic Bayesian Networks as Formal Abstractions of Structured Stochastic Processes\footnote{This work was partially supported by the European Commission IAPP project AMBI 324432, and by the John Fell OUP Research Fund.}}
\titlerunning{DBNs as Formal Abstractions of Structured Stochastic Processes} %optional, in case that the title is too long; the running title should fit into the top page column

\author[1]{Sadegh Esmaeil Zadeh Soudjani}
\author[2]{Alessandro Abate}
\author[3]{Rupak Majumdar}
\affil[1,3]{Max Planck Institute for Software Systems (MPI-SWS), Germany\\
  \texttt{\{sadegh,rupak\}@mpi-sws.org}}
\affil[2]{Department of Computer Science, University of Oxford, United Kingdom\\
  \texttt{alessandro.abate@cs.ox.ac.uk}}
\authorrunning{S. Esmaeil Zadeh Soudjani and A. Abate and R. Majumdar} %mandatory. First: Use abbreviated first/middle names. Second (only in severe cases): Use first author plus 'et. al.'

\Copyright{Sadegh Esmaeil Zadeh Soudjani and Alessandro Abate and Rupak Majumdar}%mandatory, please use full first names. LIPIcs license is "CC-BY";  http://creativecommons.org/licenses/by/3.0/

\serieslogo{}%please provide filename (without suffix)
\volumeinfo%(easychair interface)
  {Billy Editor and Bill Editors}% editors
  {2}% number of editors: 1, 2, ....
  {Conference title on which this volume is based on}% event
  {1}% volume
  {1}% issue
  {1}% starting page number
\EventShortName{}
\begin{document}

\maketitle

\begin{abstract} 
We study the problem of finite-horizon probabilistic invariance for discrete-time Markov processes over general (uncountable) state spaces.
We compute discrete-time, finite-state Markov chains as formal abstractions of general Markov processes.
Our abstraction differs from existing approaches in two ways. 
First, we exploit the structure of the underlying Markov process to compute the abstraction separately for each dimension.
Second, we employ dynamic Bayesian networks (DBN) as compact representations of the abstraction. 
In contrast, existing approaches represent and store the (exponentially large) Markov chain explicitly, 
which leads to heavy memory requirements limiting the application to models of dimension less than half, 
according to our experiments.  

We show how to construct a DBN abstraction of a Markov process satisfying an independence assumption on the driving process noise.
We compute a guaranteed bound on the error in the abstraction w.r.t.\ the probabilistic invariance property; the dimension-dependent
abstraction makes the error bounds more precise than existing approaches.
Additionally, we show how factor graphs and the sum-product algorithm for DBNs can be used to solve the finite-horizon probabilistic
invariance problem. 
Together, DBN-based representations and algorithms can be significantly more efficient than explicit 
representations of Markov chains for abstracting and model checking structured Markov processes. 
%
% 
% Discrete-time, finite-state Markov chains have been widely used as abstractions of discrete-time Markov processes 
% (stochastic dynamical systems) evolving over general state spaces. 
% In existing work the Markov chains are represented and stored explicitly: 
% in many applications the abstract model can be very large, which poses memory limitations. 
% In view of the fact that existing abstraction techniques ignore the structure of the underlying Markov processes, 
% in this paper we employ dynamic Bayesian networks (DBN) as abstractions of discrete-time Markov processes --   
% DBN have been used extensively in the learning literature as compact representations of 
% joint probability distributions. 

% \alex{[here we don't say how the abstraction is done, say unlike other approaches in the literature. no reference to error calculation/ formality (as per title) of the abstraction.]}
% Additionally, we show how \alex{inference algorithms [might induce the readers that we do actual learning?]} for DBN can be used to 
% compute the probability of the underlying Markov process to satisfy a safety property over a finite horizon
% with a guaranteed bound on the error introduced by the abstraction.
% We demonstrate that DBN-based representations and algorithms can be significantly more efficient than explicit representations of Markov chains for abstracting and model checking stochastic dynamical systems.
 \end{abstract}

\section{Introduction}
\label{sec:intro}

%\alex{Some references miss the LNCS volume and/or the publisher (Springer V.)}

Markov processes over general uncountable state spaces appear in many areas of engineering such as power networks, transportation, biological systems, robotics, and manufacturing systems.
The importance of this class of stochastic processes in applications has motivated a significant research effort into their foundations and their verification. 

We study the problem of algorithmically verifying finite-horizon probabilistic invariance for Markov processes, which is the problem of computing the probability that 
a stochastic process remains within a given set for a given finite time horizon. 
For finite-state stochastic processes, 
there is a mature theory of model checking discrete-time Markov chains \cite{BK08}, 
and a number of probabilistic model checking tools \cite{KKZ05,KNP11} 
that compute explicit solutions to the verification problem. 
On the other hand, stochastic processes taking values over uncountable state spaces may not have explicit
solutions and their numerical verification problems are undecidable even for simple dynamics \cite{APKL10}.
A number of studies have therefore explored \emph{abstraction} techniques that reduce the given stochastic process (over a general state space)
to a finite-state process, while preserving properties in a quantitative sense \cite{APKL10,SA11}.
The abstracted model allows the application of standard model checking techniques over finite-state models.  
The work in \cite{APKL10} has further shown that an explicit error can be attached to the abstraction. 
This error is computed purely based on continuity properties of the 
concrete Markov process. 
Properties proved on the finite-state abstraction can be used to reason about properties of the original system.
The overall approach has been extended to linear temporal specifications \cite{TA13} and
software tools have been developed to automate the abstraction procedure \cite{FAUST15}. 

In previous works, 
the structure of the underlying Markov process (namely, the interdependence among its variables) has not been actively reflected in the abstraction algorithms, 
and the finite-state Markov chain has been always represented explicitly,  
% and the structure of the underlying dynamical system is ignored.
which is quite expensive in terms of memory requirements.   
In many applications, 
the dynamics of the Markov process, which are characterized by a conditional kernel, often exhibit specific structural properties. 
More specifically, the dynamics of any state variable depends on only a small number of 
other state variables and the process noise driving each state variable is assumed to be independent.  
Examples of such structured systems are models of power grids and sensor-actuator networks as
large-scale interconnected networks \cite{S05} and mass-spring-damper systems \cite{AVDSMAT12,A12}.
	
We present an abstraction and model checking algorithm for discrete-time stochastic dynamical systems over general (uncountable) state spaces.
Our abstraction constructs a finite-state Markov abstraction of the process, 
but differs from previous work in that it is based on a dimension-dependent partitioning of the state space.
Additionally, we perform a precise dimension-dependent analysis of the error introduced by the abstraction, and our error bounds can
be exponentially smaller than the general bounds obtained in \cite{APKL10}. 
Furthermore, we represent the abstraction as a dynamic Bayesian network (DBN) \cite{KF09} instead of explicitly
representing the probabilistic transition matrix.
The Bayesian network representation uses independence assumptions in the model to provide 
potentially polynomial sized representations (in the number of dimensions)
for the Markov chain abstraction for which the explicit transition matrix is exponential in the dimension.
We show how factor graphs and the sum-product algorithm, developed for belief propagation in Bayesian networks, 
can be used to model check probabilistic invariance properties without constructing the transition matrix.
Overall, our approach leads to significant reduction in computational and memory resources for model checking structured
Markov processes and provides tighter error bounds.

%%%%%%%
% Algorithms for model checking of DBNs can be classified into two main categories.
% The first category includes sample-based algorithms which involve sampling trajectories from the model in an iid fashion.
% The probability of satisfying a formula,
% e.g. safety specification,
% is then \emph{estimated} by keeping track of the ratio of the number of satisfying trajectories versus
% the total number of trajectories sampled. Alternatively,
% this probability is estimated via Bayesian statistical model checking which assigns a prior distribution to it and sequentially computes its posterior distribution after each sampling.
% Hypothesis testing is also widely used for the cases where the objective is decide about comparing the satisfaction probability with a given threshold.
% The main disadvantage of these algorithms is the fact that the result is valid with a confidence: there is a chance that the outcome is not correct.
% %
% The second category comprises numerical methods and compute the \emph{exact} probability of satisfying the specification. These algorithms rely on converting the DBN to a finite-state Markov chain and doing the computation over the constructed model.
% For this purpose the transition probability matrix of the Markov chain must be computed which can be very huge for DBNs with large number of variables.
% In this article, we consider this exact computation and show that the inference algorithms for DBNs can be used to efficiently compute
% the safety probability over a finite horizon without constructing the transition matrix.

The material is organized in six sections.  
%After this brief introduction, 
Section \ref{sec:prelim} defines discrete-time Markov processes 
and the probabilistic invariance problem. 
Section~\ref{sec:DBN_abstraction} presents a new algorithm for 
abstracting a process to a DBN, 
together with the quantification of the abstraction error. 
We discuss efficient model checking of the constructed DBN in Section \ref{sec:eff_comp}, 
and apply the overall abstraction algorithm to a case study in Section \ref{sec:case_study}. 
Section \ref{sec:concl} outlines some further directions of investigation. 
Proofs of statements are included in the Appendix. 
% to prevent interruption in the general flow of the paper.

\section{Markov Processes and Probabilistic Invariance}
\label{sec:prelim}

% We formally define discrete-time Markov processes and discuss their equivalence to stochastic dynamical systems.
% The problem of finite-horizon probabilistic invariance defined over these processes is also presented.

\subsection{Discrete-Time Markov Processes}
\label{subsec:dt-MP}

We write $\mathbb N$ for the non-negative integers $\mathbb N =\{0,1,2,\ldots\}$ and $\mathbb N_n = \{1,2,\ldots,n\}$.
We use bold typeset for vectors and normal typeset for one-dimensional quantities.

We consider a discrete-time Markov process $\mathscr M_{\mathfrak s}$ defined over a general state space, 
and characterized by the tuple $(\mathcal S,\mathcal B, T_{\mathfrak s})$: 
$\mathcal S$ is the continuous state space, which we assume to be endowed with a metric and to be separable\footnote{
A metric space $\mathcal S$ is called separable if it has a countable dense subset.}; 
$\mathcal B$ is the Borel $\sigma$-algebra associated to $\mathcal S$, 
which is the smallest $\sigma$-algebra containing all open subsets of $\mathcal S$; 
%For the separable metric space $\mathcal S$, $\mathcal B$ equals the $\sigma$-algebra generated by the open (or closed) balls of $\mathcal S$.
and $T_{\mathfrak s}:\mathcal S\times\mathcal B\rightarrow[0,1]$ is a stochastic kernel, so that
$T_{\mathfrak s}(\cdot,B)$ is a non-negative measurable function for any set $B\in\mathcal B$, and
$T_{\mathfrak s}(\vr s,\cdot)$ is a probability measure on $(\mathcal S,\mathcal B)$ for any $\vr s\in\mathcal S$. 
Trajectories (also called traces or paths) of $\mathscr M_{\mathfrak s}$ are sequences 
$(\vr s(0),\vr s(1),\vr s(2),\ldots)$ which belong to the set $\Omega = \mathcal S^{\mathbb N}$.
The product $\sigma$-algebra on $\Omega$ is denoted by $\mathcal F$. 
Given the initial state $\vr s(0) = \vr s_0\in\mathcal S$ of $\mathscr M_{\mathfrak s}$, 
the stochastic Kernel $T_{\mathfrak s}$ induces a unique probability measure $\mathcal P$ on $(\Omega,\mathcal F)$ that satisfies the Markov property: 
namely for any measurable set $B\in\mathcal B$ and any $t \in\mathbb N$
\begin{equation*}
	\mathcal P \left(\vr s(t+1)\in\ B| \vr s(0),\vr s(1),\ldots,\vr s(t)\right) = \mathcal P \left(\vr s(t+1)\in\ B| \vr s(t)\right) = T_{\mathfrak s}(\vr s(t),B).
\end{equation*}
We assume that the stochastic kernel $T_{\mathfrak s}$ admits a density function $t_{\mathfrak s}:\mathcal S\times\mathcal S\rightarrow\mathbb R_{\ge 0}$,
such that $T_{\mathfrak s}(\vr s,B) = \int_B t_{\mathfrak s}(\bar{\vr s}|\vr s)d\bar{\vr s}$.

\smallskip

A familiar class of discrete-time Markov processes is that of stochastic dynamical systems. If $\{\vr \zeta(t),\,\,t\in\mathbb N\}$ is a sequence of independent and identically distributed (iid) random variables taking values in $\mathbb R^n$, and $\vr f:\mathcal S\times \mathbb R^n\rightarrow \mathcal S$ is a measurable map, then the recursive equation 
\begin{equation}
	\label{eq:dyn_system}
	\vr s(t+1) = \vr f(\vr s(t),\vr \zeta(t)),\quad\forall t\in\mathbb N,\quad \vr s(0) = \vr s_0\in\mathcal S,
\end{equation}
induces a Markov process that is characterized by the kernel
\begin{equation*}
T_{\mathfrak s}(\vr s,B) = T_{\zeta}\left(\vr \zeta\in\mathbb R^n\,:\,\vr f(\vr s,\vr \zeta)\in B\right),
\end{equation*}
where $T_{\zeta}$ is the distribution of the r.v. $\vr \zeta(0)$ (in fact, of any $\vr \zeta(t)$ since these are iid random variables). 
In other words, the map $\vr f$ together with the distribution of the r.v. $\{\vr \zeta(t)\}$ uniquely define the stochastic kernel of the process. 
The converse is also true as shown in \cite[Proposition 7.6]{k2002}: any discrete-time Markov process $\mathscr M_{\mathfrak s}$ admits a dynamical representation as in \eqref{eq:dyn_system}, 
for an appropriate selection of function $\vr f$ and distribution of the r.v. $\{\vr \zeta(t)\}$. 

% \smallskip 

Let us expand the dynamical equation \eqref{eq:dyn_system} explicitly over its states $\vr s = [s_1,\ldots,s_n]^T$, 
map components $\vr f = [f_1,\ldots,f_n]^T$, 
and uncertainly terms $\vr \zeta = [\zeta_1,\ldots,\zeta_n]^T$, as follows: 
\begin{equation}
	\label{eq:dyn_system_expand}
	\begin{array}{l}
		s_1(t+1) = f_1(s_1(t),s_2(t),\ldots,s_n(t),\zeta_1(t)),\\
		s_2(t+1) = f_2(s_1(t),s_2(t),\ldots,s_n(t),\zeta_2(t)),\\
		\quad \vdots\\
		s_n(t+1) = f_n(s_1(t),s_2(t),\ldots,s_n(t),\zeta_n(t)).
	\end{array}
\end{equation}
In this article we are interested in exploiting the knowledge of the structure of the dynamics in \eqref{eq:dyn_system_expand} for formal verification via abstractions \cite{APKL10,SA11,SA13}.
%\alex{[ref to EJC, SIAM articles]}.  
We focus our attention to continuous (unbounded and uncountable) Euclidean spaces $\mathcal S = \mathbb R^n$, 
and further assume that for any $t\in\mathbb N$,  $\zeta_k(t)$ are independent for all $k\in\mathbb N_n$. 
This latter assumption is widely used in the theory of dynamical systems, 
and allows for the following multiplicative structure on the conditional density function of the process: 
\begin{equation}
\label{eq:prod_structure}
	t_{\mathfrak s}(\bar{\vr s}|\vr s) = t_1(\bar s_1|\vr s)t_2(\bar s_2|\vr s)\ldots t_n(\bar s_n|\vr s),
\end{equation}
where the function $t_k:\mathbb R^n\times\mathbb R\rightarrow\mathbb R_{\ge 0}$
solely depends on the map $f_k$ and the distribution of $\zeta_k$. 
The reader is referred to Section \ref{sec:case_study} for the detailed computation of the functions $t_k$ from the dynamical equations in \eqref{eq:dyn_system_expand}.

\begin{remark} 
The results of this article are presented under the structural assumption that $\zeta_k(\cdot)$ are independent over $k\in\mathbb N_n$. 
These results can be generalized to a broader class of processes by allowing inter-dependencies between the entries of the process noise, 
which requires 
partitioning the set of entries of $\vr\zeta(\cdot)$ so that any two entries from different partition sets are independent, 
whereas entries within a partition set may still be dependent. 
This assumption induces a multiplicative structure on $t_{\mathfrak s}(\bar{\vr s}|\vr s)$ with respect to the partition, 
which is similar to \eqref{eq:prod_structure}.
The finer the partition, the more efficient is our abstraction process.
% The presented approach is more efficient for systems with more partition sets (optimally, as many as the model variables). 
\end{remark}

\begin{example}
	Figure~\ref{fig:exm_msp} shows a system of $n$ masses connected by springs and dampers.
	For $i\in\mathbb N_n$, block $i$ has mass $m_i$,
	the $i^{\text{th}}$ spring has stiffness $k_i$, 
	and the $i^{\text{th}}$ damper has damping coefficient $b_i$. 
	The first mass is connected to a fixed wall by the left-most spring/damper connection. 
	All other masses are connected to the previous mass with a spring and a damper. 
	A force $\zeta_i$ is applied to each mass, modeling the effect of a disturbance or of process noise. 
	The dynamics of the overall system is comprised of the position and velocity of the blocks. 
	It can be shown that the dynamics in discrete time take the form $\vr s(t+1) = \Phi \vr s(t)+\vr \zeta(t)$, where $\vr s(t)\in\mathbb R^{2n}$ with $s_{2i-1}(t),s_{2i}(t)$ indicating the velocity and position of mass $i$. 
	The state transition matrix $\Phi = [\Phi_{ij}]_{i,j}\in\mathbb R^{2n\times 2n}$ is a band matrix with lower and upper bandwidth $3$ and $2$, respectively ($\Phi_{ij} = 0$ for $j<i-3$ and for $j>i+2$). \qed
\end{example}
\begin{example}
	A second example of structured dynamical systems is a discrete-time large-scale interconnected system.  
	Consider an interconnected system of $N_{\mathfrak d}$ heterogeneous linear time-invariant (LTI) subsystems described by the following stochastic difference equations: 
	\begin{align*}
	\vr s_i(t+1) = \Phi_i \vr s_i(t) + \sum_{j\in N_i} G_{ij}\vr s_j(t) + B_i \vr u_i(t)+\vr\zeta_i(t),
	\end{align*}
	where $i\in\mathbb N_{N_{\mathfrak d}}$ denotes the $i^{\text{th}}$ subsystem and $\vr s_i\in\mathbb R^{n\times 1}, \vr u_i\in\mathbb R^{p\times 1}, \vr\zeta_i\in\mathbb R^{m\times 1}$ are the state, the input, and the process noise of subsystem $i$. The term $\sum_{j\in N_i} G_{ij}\vr s_j(t)$ represents the physical interconnection between the subsystems where $N_i$, $|N_i|\ll N_{\mathfrak d}$, is the set of subsystems to which system $i$ is physically connected. The described interconnected system can be found in many application areas including smart power grids, traffic systems, and sensor-actuator networks \cite{GH11}. \qed
	%http://mediatum.ub.tum.de/doc/1082135/55113.pdf
\end{example}

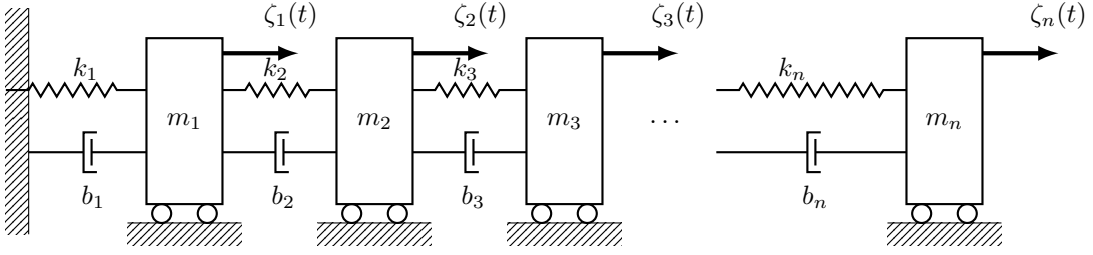
\begin{figure}
	\begin{tikzpicture}[every node/.style={draw,outer sep=0pt,thick}]
	\tikzstyle{spring}=[thick,decorate,decoration={zigzag,pre length=0.3cm,post length=0.3cm,segment length=6}]
	\tikzstyle{damper}=[thick,decoration={markings,  
		mark connection node=dmp,
		mark=at position 0.5 with 
		{
			\node (dmp) [thick,inner sep=0pt,transform shape,rotate=-90,minimum width=15pt,minimum height=3pt,draw=none] {};
			\draw [thick] ($(dmp.north east)+(2pt,0)$) -- (dmp.south east) -- (dmp.south west) -- ($(dmp.north west)+(2pt,0)$);
			\draw [thick] ($(dmp.north)+(0,-5pt)$) -- ($(dmp.north)+(0,5pt)$);
		}
	}, decorate]
	\tikzstyle{ground}=[fill,pattern=north east lines,draw=none,minimum width=0.75cm,minimum height=0.3cm]
	
	\begin{scope}[xshift=3cm]
	\node (M1) [minimum width=1cm, minimum height=2.2cm] {$m_1$};
	\node (ground1) [ground,anchor=north,yshift=-0.25cm,minimum width=1.5cm] at (M1.south) {};
	\draw (ground1.north east) -- (ground1.north west);
	\draw [thick] (M1.south west) ++ (0.2cm,-0.125cm) circle (0.125cm)  (M1.south east) ++ (-0.2cm,-0.125cm) circle (0.125cm);
	\node (wall) [ground, rotate=-90, minimum width=3cm,yshift=-2.2cm] {};
	\draw (wall.north east) -- (wall.north west);
	\draw [spring] (wall.200) -- ($(M1.north west)!(wall.200)!(M1.south west)$);
	\draw [damper] (wall.20) -- ($(M1.north west)!(wall.20)!(M1.south west)$);
	\node (consd1) [draw=none,fill=none,xshift=-1.2cm,yshift=-1cm] {$b_1$};
	\node (consd1) [draw=none,fill=none,xshift=-1.3cm,yshift=0.7cm] {$k_1$};
	\draw [-latex,ultra thick] (M1.north east) ++ (0,-0.2cm) -- +(1cm,0); % for forces
	\end{scope}
	\begin{scope}[xshift=5.5cm]
	\node (M2) [minimum width=1cm, minimum height=2.2cm] {$m_2$};
	\node (ground2) [ground,anchor=north,yshift=-0.25cm,minimum width=1.5cm] at (M2.south) {};
	\draw (ground2.north east) -- (ground2.north west);
	\draw [thick] (M2.south west) ++ (0.2cm,-0.125cm) circle (0.125cm)  (M2.south east) ++ (-0.2cm,-0.125cm) circle (0.125cm);
	\draw [spring] ($(M1.north east)!(wall.200)!(M1.south east)$) -- ($(M2.north west)!(wall.200)!(M2.south west)$);
	\draw [damper] ($(M1.north east)!(wall.20)!(M1.south east)$) -- ($(M2.north west)!(wall.20)!(M2.south west)$);
	\node (consd1) [draw=none,fill=none,xshift=-1.2cm,yshift=-1cm] {$b_2$};
	\node (consd1) [draw=none,fill=none,xshift=-1.3cm,yshift=0.7cm] {$k_2$};
	\node (consd1) [draw=none,fill=none,xshift=-1.1cm,yshift=1.4cm] {$\zeta_1(t)$};
	\draw [-latex,ultra thick] (M2.north east) ++ (0,-0.2cm) -- +(1cm,0);
	\end{scope}
	\begin{scope}[xshift=8cm]
	\node (M3) [minimum width=1cm, minimum height=2.2cm] {$m_3$};
	\node (ground3) [ground,anchor=north,yshift=-0.25cm,minimum width=1.5cm] at (M3.south) {};
	\draw (ground3.north east) -- (ground3.north west);
	\draw [thick] (M3.south west) ++ (0.2cm,-0.125cm) circle (0.125cm)  (M3.south east) ++ (-0.2cm,-0.125cm) circle (0.125cm);
	\draw [spring] ($(M2.north east)!(wall.200)!(M2.south east)$) -- ($(M3.north west)!(wall.200)!(M3.south west)$);
	\draw [damper] ($(M2.north east)!(wall.20)!(M2.south east)$) -- ($(M3.north west)!(wall.20)!(M3.south west)$);
	\node (consd1) [draw=none,fill=none,xshift=-1.2cm,yshift=-1cm] {$b_3$};
	\node (consd1) [draw=none,fill=none,xshift=-1.3cm,yshift=0.7cm] {$k_3$};
	\node (consd1) [draw=none,fill=none,xshift=-1.1cm,yshift=1.4cm] {$\zeta_2(t)$};
	\draw [-latex,ultra thick] (M3.north east) ++ (0,-0.2cm) -- +(1cm,0);
	\node (contin) [draw=none,fill=none,xshift=1.5cm] {$\ldots\quad$};
	\end{scope}
	\begin{scope}[xshift=13cm]
	\node (Mn) [minimum width=1cm, minimum height=2.2cm] {$m_n$};
	\node (groundn) [ground,anchor=north,yshift=-0.25cm,minimum width=1.5cm] at (Mn.south) {};
	\draw (groundn.north east) -- (groundn.north west);
	\draw [thick] (Mn.south west) ++ (0.2cm,-0.125cm) circle (0.125cm)  (Mn.south east) ++ (-0.2cm,-0.125cm) circle (0.125cm);	
	\draw [spring] ($(contin.north east)!(wall.200)!(contin.south east)$) -- ($(Mn.north west)!(wall.200)!(Mn.south west)$);
	\draw [damper] ($(contin.north east)!(wall.20)!(contin.south east)$) -- ($(Mn.north west)!(wall.20)!(Mn.south west)$);
	\node (consd1) [draw=none,fill=none,xshift=-1.7cm,yshift=-1cm] {$b_n$};
	\node (consd1) [draw=none,fill=none,xshift=-2cm,yshift=0.7cm] {$k_n$};
	\node (consd1) [draw=none,fill=none,xshift=-3.5cm,yshift=1.4cm] {$\zeta_{3}(t)$};
	\node (consd1) [draw=none,fill=none,xshift=1.5cm,yshift=1.4cm] {$\zeta_n(t)$};
	\draw [-latex,ultra thick] (Mn.north east) ++ (0,-0.2cm) -- +(1cm,0);
	\end{scope}
	\end{tikzpicture}
	\caption{$n$-body mass-spring-damper system.}
	\label{fig:exm_msp}
\end{figure}

\subsection{Probabilistic Invariance}
\label{subsec:safety}

% We are interested in verifying properties and proving correctness of the stochastic systems described in the previous section.
We focus on verifying probabilistic invariance, which 
plays a central role in verifying properties of a system expressed as 
PCTL formulae or as linear temporal specifications \cite{BK08,rcsl2010,TA13}.
% The abstraction algorithm of this paper is tailored to this fundamental invariance property, which 
% can as well be employed for abstracting the system regardless of the specifics of the probabilistic 
% invariance problem by considering the state space as the invariance set \cite{FAUST15,TA13}.
%\alex{why is prob. invariance important? why not generating abstractions and feeding them directly to a PMC for further verification? we need to motivate this section better.}

% The problem of finite-horizon probabilistic invariance (alternatively referred to as probabilistic safety) can be formalised as follows:
\begin{definition}[Probabilistic Invariance]
Consider a bounded Borel set $A\in\mathcal B$, representing a set of safe states.
The finite-horizon \emph{probabilistic invariance problem} asks to compute the probability that a trajectory of $\mathscr M_{\mathfrak s}$ associated with an initial condition $\vr s_0$
remains within the set $A$ during the finite time horizon $N$:
\begin{equation*}
	p_N(\vr s_0,A) = \mathcal P\{\vr s(t)\in A\text{ for all } t=0,1,2,\ldots,N| \vr s(0) =\vr s_0\}.
\end{equation*}
%\qed
\end{definition}

This quantity allows us to extend the result to a general probability distribution $\pi:\mathcal B\rightarrow [0,1]$ for the initial state $\vr s(0)$ of the system as
\begin{equation}
\label{eq:initial_dis}
\mathcal P\{\vr s(t)\in A\text{ for all } t=0,1,2,\ldots,N\} = \int_{\mathcal S} p_N(\vr s_0,A)\pi(d\vr s_0).	
\end{equation}
Solution of the probabilistic invariance problem can be characterized via the value functions $V_k:\mathcal S\rightarrow[0,1]$, $k=0,1,2,\ldots,N$, defined by the following Bellman backward recursion \cite{APKL10}:
%\alex{[indexing in previous equation could be improved, it is now counter-intuitive.]}
\begin{equation}
	\label{eq:bellman_rec}
	V_k(\vr s) = \vr 1_A(\vr s) \int_A V_{k+1}(\bar{\vr s})t_{\mathfrak s}(\bar{\vr s}|\vr s)d\bar{\vr s}\,\,\text{ for }\,\,k=0,1,2,\ldots,N-1.
\end{equation}
This recursion is initialized with $V_N(\vr s) = \vr 1_A(\vr s)$, where $\vr 1_A(\vr s)$ is the indicator function which is $1$ if $\vr s\in A$ and $0$ otherwise,
% \begin{equation*}
% 	V_N(\vr s) = \vr 1_A(\vr s) = 
% 	\begin{cases}
% 		1 & \text{ if } \vr s\in A\\
% 		0 & \text{ otherwise,}
% 	\end{cases}
% \end{equation*}
and results in the solution $p_N(\vr s_0,A) = V_0(\vr s_0)$.

Equation \eqref{eq:bellman_rec} characterizes the finite-horizon probabilistic invariance quantity as the solution
of a dynamic programming problem. However, since its explicit solution is in general not available, the actual computation of the quantity $p_N(\vr s_0,A)$ requires $N$ numerical integrations at each state in the set $A$.
This is usually performed with techniques based on state-space discretization \cite{B75}.

\section{Formal Abstractions as Dynamic Bayesian Networks}
\label{sec:DBN_abstraction}

% In this section we define Dynamic Bayesian Networks and utilize the structure of the stochastic system to model it as a DBN with continuous random variables.
% Then we present an algorithm to abstract it as a DBN with discrete random variables in order to facilitate formal verification of the stochastic system.
% The abstraction error is quantified formally by putting continuity assumptions on the conditional density functions of the systems.

\subsection{Dynamic Bayesian Networks}
\label{subsec:DBN}

%\alex{As of now, it is unclear whether DBn are used simply as a re-formulation of dtMP, or as a framework for their finite-state abstraction. We should clarify this throughout. For example, in the definition of BN below are the rv discrete or continuous?}

A Bayesian network (BN) is a tuple $\mathfrak B = (\mathcal V,\mathcal E,\mathcal T)$. The pair 
$(\mathcal V,\mathcal E)$ is a directed Acyclic Graph (DAG) representing the structure of the network. 
The nodes in $\mathcal V$ are (discrete or continuous) random variables and the arcs in $\mathcal E$ represent the dependence relationships among the random variables. 
The set $\mathcal T$ contains conditional probability distributions (CPD) in forms of tables or density functions for discrete and continuous random variables, respectively.
In a BN, knowledge is represented in two ways: 
qualitatively, as dependences between variables by means of the DAG; and 
quantitatively, as conditional probability distributions attached to the dependence relationships. 
Each random variable $X_i\in\mathcal V$ is associated with a conditional probability distribution $\mathbb P(X_i|Pa(X_i))$,
where $Pa(Y)$ represents the parent set of the variable $Y\in\mathcal V$: $Pa(Y) = \{X\in\mathcal V|(X,Y)\in\mathcal E\}$.
A BN is called \emph{two-layered} if the set of nodes $\mathcal V$ can be partitioned to two 
sets $\mathcal V_1,\mathcal V_2$ with the same cardinality such that only the nodes in the second layer $\mathcal V_2$ have an associated CPD.

A dynamic Bayesian network \cite{KF09,Murphy02} is a way to extend Bayesian networks to model probability distributions over collections of random variables $X(0),X(1),X(2),\ldots$ indexed by time $t$.
A DBN\footnote{The DBNs considered in this paper
%have a \alex{time-homogeneous} structure (namely, the two-layered Bayesian networks do not vary with the time index $t$)
are stationary (the structure of the network does not change with the time index $t$).
They have no input variables and are fully observable: the output of the DBN model equals to its state.} 
is defined to be a pair $(\mathfrak B_0,\mathfrak B_{\rightarrow})$, where $\mathfrak B_0$ is a BN which defines the distribution of $X(0)$, 
and $\mathfrak B_{\rightarrow}$ is a two-layered BN that defines the transition probability distribution for $(X(t+1)|X(t))$.  

\subsection{DBNs as Representations of Markov Processes}
\label{subsec:cDBN}

We now show that any discrete-time Markov process $\mathscr M_{\mathfrak s}$ over $\mathbb R^n$ can be represented as a DBN $(\mathfrak B_0, \mathfrak B_{\rightarrow})$ over
$n$ continuous random variables.
The advantage of the reformulation is that it makes the dependencies between random variables explicit.
% In the following we describe how to construct $\mathfrak B_0$ and $\mathfrak B_{\rightarrow}$ in order to reformulate the discrete-time Markov process $\mathscr M_{\mathfrak s}$ as a DBN. 

The BN $\mathfrak B_0$ is trivial for a given initial state of the Markov process $\vr s(0) = \vr s_0$.
The DAG of $\mathfrak B_0$ has the set of nodes $\{X_1,X_2,\ldots,X_n\}$ without any arc.
The Dirac delta distribution located in the initial state of the process is assigned to each node of $\mathfrak B_0$.\footnote{
	For a general initial probability distribution $\pi:\mathcal B\rightarrow[0,1]$, a set of arcs must be added to 
	reflect its possible product structure. 
	This construction is not important at the current stage because of the backward recursion formulation of the probabilistic safety
	(please refer to \eqref{eq:initial_dis} in Section \ref{subsec:safety}).}
%we first construct the DAG and the associated conditional probability distributions.
%In the following we first construct the graph and then represent computation of the conditional distributions.
The DAG for the two-layered BN $\mathfrak B_{\rightarrow} = (\mathcal V,\mathcal E,\mathcal T)$ comprises a set of nodes
$\mathcal V = \mathcal V_1\cup\mathcal V_2$, 
with $\mathcal V_1 = \{X_1,X_2,\ldots,X_n\}$ and $\mathcal V_2 = \{\bar X_1,\bar X_2,\ldots,\bar X_n\}$.
Each arc in $\mathcal E$ connects a node in $\mathcal V_1$ to another node in $\mathcal V_2$;  
$(X_i,\bar X_j)\in\mathcal E$ if and only if $t_j(\bar s_j|\vr s)$ is not a constant function of $s_i$.
The set $\mathcal T$ assigns a CPD to each node $\bar X_j$ according to the density function $t_j(\bar s_j|\vr s)$.
%
% Example \ref{exm:linear_system} illustrates this construction.

\begin{example}
	\label{exm:linear_system}
	Consider the following stochastic linear dynamical system: 
	\begin{equation}
		\label{eq:exmp_linear}
		\vr s(t+1) = \Phi \vr s(t)+\vr \zeta(t)\quad t\in\mathbb N,\quad \vr s(0) = \vr s_0 = [s_{01},s_{02},\ldots,s_{0n}]^T,
	\end{equation}
	where $\Phi = [a_{ij}]_{i,j}$ is the system matrix and $\vr \zeta(t)\sim\mathcal N(0,\Sigma)$ are independent Gaussian r.v. for any $t\in\mathbb N$.
	The covariance matrix $\Sigma$ is assumed to be full rank. Consequently, a linear transformation can be employed to change the coordinates and obtain a stochastic linear system with a
	%\alex{[hint at why -- recall assumption made in previous section]} we assume that $\Sigma$ is 
	diagonal covariance matrix. Then without loss of generality we assume $\Sigma = diag([\sigma_1^2,\sigma_2^2,\ldots,\sigma_n^2])$, which clearly satisfies the independence assumption on the process noise raised in Section \ref{subsec:dt-MP}. 
	Model \eqref{eq:exmp_linear} for a lower bidiagonal matrix $\Phi$ can be expanded as follows: 
	\begin{align*}
		& s_1(t+1) = a_{11} s_1(t) + \zeta_1(t)\\
		& s_2(t+1) = a_{21} s_1(t) + a_{22} s_2(t) + \zeta_2(t)\\
		& s_3(t+1) = a_{32} s_2(t) + a_{33} s_3(t) + \zeta_3(t)\\
		& \hspace{0.3in} \vdots\\
		& s_n(t+1) = a_{n(n-1)} s_{n-1}(t) + a_{nn} s_n(t) + \zeta_n(t),
	\end{align*}
	where $\zeta_i(\cdot),\,i\in\mathbb N_n$ are independent Gaussian r.v. $\mathcal N(0,\sigma_i^2)$.
	The conditional density function of the system takes the following form: 
	\begin{equation*}
		t_{\mathfrak s}(\bar{\vr s}|\vr s) = t_1(\bar s_1|s_1)t_2(\bar s_2|s_1,s_2)t_3(\bar s_3|s_2,s_3)\ldots t_n(\bar s_n|s_{n-1},s_n).
	\end{equation*}
	The DAG of the two-layered BN $\mathfrak B_{\rightarrow}$ associated with this system is sketched in Figure \ref{fig:DBN_graph} for $n = 4$.
	\begin{figure}
		\centering
		\includegraphics{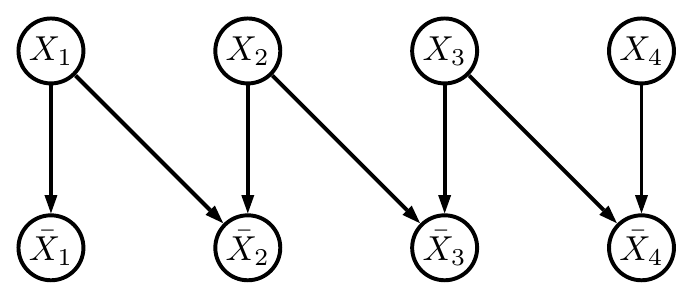}
		\caption{Two-layered BN $\mathfrak B_{\rightarrow}$ associated with the stochastic linear dynamical system in \eqref{eq:exmp_linear} for $n=4$.}
		\label{fig:DBN_graph}
	\end{figure} 
	The BN $\mathfrak B_0$ has an empty graph on the set of nodes $\{X_1,\ldots,X_n\}$ with the associated Dirac delta density functions located at $s_{0i}$, $\delta_d(s_i(0)-s_{0i})$.
	\qed
\end{example}

%The constructed DAG of $\mathfrak B_{\rightarrow}$ is a directed balanced bipartite graph $(\mathcal V_1,\mathcal V_2,\mathcal E)$ (the graph has two sets of node with the same cardinality and there is no arc connecting nodes of any of these sets). 
%The \alex{bi}adjacency \alex{why bi-?} matrix of this graph is a $n\times n$ matrix $M = [M_{ij}]_{i,j}$, 
%with $M_{ij} = 1$ if $(X_i,\bar X_j)\in\mathcal E$, 
%and $M_{ij} = 0$ otherwise. 
%This matrix is used in Section \ref{sec:eff_comp} for model checking purposes. 
%The biadjacency matrix of the system in Example \ref{exm:linear_system} is
%\begin{equation*}
%	M = \left[
%	\begin{array}{cccc}
%	1 & 1 & 0 & 0\\
%	0 & 1 & 1 & 0\\
%	0 & 0 & 1 & 1\\
%	0 & 0 & 0 & 1
%	\end{array}	
%	\right].
%\end{equation*}

\subsection{Finite Abstraction of Markov Processes as Discrete DBNs}
\label{subsec:CPD}

Let $A\in \mathcal B$ be a bounded Borel set of safe states.
We abstract the structured Markov process $\mathscr M_{\mathfrak s}$
interpreted in the previous section as a DBN with continuous variables
to a DBN with discrete random variables. 
Our abstraction is relative to the set $A$.
%\alex{[Here it is unclear whether you're abstracting the dtMP model per-se, or around the safety problem introduced before. We need to clarify this.]}
%
Algorithm \ref{algo:DBN_app} provides the steps of the abstraction procedure.
It consists of discretizing each dimension into a finite number of bins.

In Algorithm \ref{algo:DBN_app},
the projection operators $\Pi_i:\mathbb R^n\rightarrow\mathbb R,\,i\in\mathbb N_n,$ are defined as $\Pi_i(\vr s) = s_i$ for any $\vr s = [s_1,\ldots,s_n]^T\in\mathbb R^n$.  
These operators are used to project the safe set $A$ over different dimensions, $D_i \doteq \Pi_i(A)$. 
In step \ref{alg:step_partition} of the Algorithm, 
set $D_i$ is partitioned as $\{D_{ij}\}_{j=1}^{n_i}$ (for any $i\in\mathbb N_n$, $D_{ij}$'s are arbitrary but non-empty, non-intersecting, and  $D_i= \cup_{j=1}^{n_i} D_{ij}$). 
%The partitions can be selected without any restriction apart from being a partition.
The corresponding representative points $z_{ij} \in D_{ij}$ are also chosen arbitrarily. 
%
%In Algorithm \ref{algo:DBN_app}, 
Step \ref{alg_step_supp} of the algorithm constructs the support of the random variables in $\mathfrak B_{\rightarrow}$, $\mathcal V = \{X_i,\bar X_i,i\in\mathbb N_n\}$, 
and step \ref{alg:step_CPD} computes the discrete CPDs $T_i(\bar X_i|Pa(\bar X_i))$, 
reflecting the dependencies among the variables.  
For any $i\in\mathbb N_n$,
$\Xi_i: Z_i \rightarrow 2^{D_i}$ represents a set-valued map that associates to any point $z_{ij}\in Z_i$ 
the corresponding partition set $D_{ij} \subset D_i$ (this is known as the ``refinement map'').
Furthermore, the abstraction map $\xi_i: D_i \rightarrow Z_i$ associates to any point $s_i \in D_i$   
the corresponding discrete state in $Z_i$. 
Additionally, notice that the absorbing states $\phi = \{\phi_1,\ldots,\phi_n\}$ are added to the definition of BN $\mathfrak B_{\rightarrow}$
so that the conditional probabilities $T_i(\bar X_i|Pa(\bar X_i))$ marginalize to one.   
The function $v(\cdot)$ used in step \ref{alg:step_CPD} acts on (possibly a set of) random variables and provides their instantiation. 
In other words, 
the term $v(Pa(\bar X_i))$ that is present in the conditioned argument of $t_i$ leads to evaluate function $t_i(\bar s_i|\cdot)$ at the instantiated values of $Pa(\bar X_i)$.  
\begin{algorithm}[t]
	\caption{Abstraction of model $\mathscr M_{\mathfrak s}$ as a DBN with $\mathfrak B_{\rightarrow} = (\mathcal V,\mathcal E,\mathcal T)$ over discrete r.v.}
	\label{algo:DBN_app}
	\begin{center}
		\begin{algorithmic}[1]
			\REQUIRE 
			input model $\mathscr M_{\mathfrak s} = (\mathcal S,\mathcal B, T_{\mathfrak s})$, safe set $A$
			\STATE
			Project safe set $A$ in each dimension $D_i \doteq \Pi_i(A),\,i\in\mathbb N_n$
			\STATE
			\label{alg:step_partition}
			Select finite $n_i$-dimensional partition of $D_i$ as $D_i= \cup_{j=1}^{n_i} D_{ij},\, i\in\mathbb N_n$ 
			%(for any $i$, $\{D_{ij}\}_j$ are non-overlapping)
			\STATE
			For each $D_{ij}$, select single representative point $z_{ij} \in D_{ij},\, z_{ij} = \xi_i (D_{ij})$
			%\alex{[the abstraction map hasn't been define anywhere]}	 
			\STATE
			Construct the DAG $(\mathcal V,\mathcal E)$, 
			with $\mathcal V = \{X_i,\bar X_i,i\in\mathbb N_n\}$ and $\mathcal E$ as per Section \ref{subsec:cDBN}  
			\STATE
			\label{alg_step_supp}
			Define $Z_i = \{z_{i1},\ldots,z_{in_i}\}$, $i\in\mathbb N_n$, and take $\Omega_i = Z_i \cup \{\phi_i\}$ as the finite state space of two r.v. $X_i$ and $\bar X_i$,
			$\phi_i$ being dummy variables as per Section \ref{subsec:CPD} 
			\STATE
			\label{alg:step_CPD}
			Compute elements of the set $\mathcal T$, namely CPD $T_i$ related to the node $\bar X_i$, $i\in\mathbb N_i$, as
			%containing the CPD $\mathbb P(\bar X_i|Pa(\bar X_i))$, $i\in\mathbb N_i$, as: 
			%\alex{[I'm not sure that the functions val have been formally defined.]}
			\begin{equation*}
				T_i(\bar X_i = z|v(Pa(\bar X_i))) 
				=
				\left\{
				\begin{array}{ll}
					\int_{\Xi_i(z)} t_i(\bar s_i| v(Pa(\bar X_i)))d\bar s_i, & z\in Z_i,\,\, v(Pa(\bar X_i))\cap\phi = \emptyset\\
					1-\sum\limits_{z\in Z_i}\int_{\Xi_i(z)} t_i(\bar s_i| v(Pa(\bar X_i)))d\bar s_i, & z=\phi_i, \,\, v(Pa(\bar X_i))\cap\phi = \emptyset\\
					1, & z = \phi_i,\,\, v(Pa(\bar X_i))\cap\phi \neq \emptyset\\
					0, & z\in Z_i,\,\, v(Pa(\bar X_i))\cap\phi \neq \emptyset
				\end{array}
				\right.
			\end{equation*}  
			\ENSURE
			output DBN with $\mathfrak B_{\rightarrow} = (\mathcal V,\mathcal E,\mathcal T)$ over discrete r.v.
		\end{algorithmic}
	\end{center}
\end{algorithm}

The construction of the DBN with discrete r.v.\ in Algorithm \ref{algo:DBN_app} is closely related to the Markov chain abstraction method in \cite{APKL10,SA13}. 
The main difference lies in partitioning in each dimension separately instead of doing it for the whole state space. 
Absorbing states are also assigned to each dimension separately instead of having only one for the unsafe set. 
Moreover, Algorithm \ref{algo:DBN_app} stores the transition probabilities efficiently as a BN.

\subsection{Probabilistic Invariance for the Abstract DBN}
\label{subsec:MC_BN}

We extend the use of $\mathbb P$ by denoting the probability measure on the set of events defined over
a DBN with discrete r.v.\ $\vr z = (X_1,X_2,\ldots,X_n)$.  
Given a discrete set $Z_{\mathfrak a}\subset \prod_i\Omega_i$, 
the probabilistic invariance problem asks to evaluate the probability $p_N(\vr z_0, Z_{\mathfrak a})$ that a finite execution associated
with the initial condition $\vr z(0) = \vr z_0$ remains within the set $Z_{\mathfrak a}$ during the finite time
horizon $t=0,1,2,\ldots,N$.
Formally, 
\begin{equation*}
p_N(\vr z_0,Z_{\mathfrak a}) = \mathbb P(\vr z(t)\in Z_{\mathfrak a}, \text{ for all } t=0,1,2,\ldots,N|\vr z(0) = \vr z_0).
\end{equation*}
This probability can be computed by a discrete analogue of the Bellman backward recursion (see \cite{APLS08} for details).

\begin{theorem}
	\label{thm:safety_dis}
	Consider value functions $V_k^d:\prod_i\Omega_i\rightarrow[0,1]$, $k=0,1,2,\ldots,N$, computed by the backward recursion
	\begin{equation}
	\label{eq:bellman_dis}
	V_k^d(\vr z) = \vr 1_{Z_{\mathfrak a}}(\vr z) \sum_{\bar{\vr z}\in\prod_i\Omega_i} V_{k+1}^d(\bar{\vr z})\mathbb P(\bar{\vr z}|\vr z)
	\quad k=0,1,2,\ldots,N-1,
	\end{equation}
	and initialized with $V_N^d(\vr z) = \vr 1_{Z_{\mathfrak a}}(\vr z)$.
	Then the solution of the invariance problem is characterized as $p_N(\vr z_0,Z_{\mathfrak a}) = V_0^d(\vr z_0)$.
\end{theorem}
The discrete transition probabilities $\mathbb P(\bar{\vr z}|\vr z)$ in Equation~\eqref{eq:bellman_dis} are computed by taking the product of the CPD in $\mathcal T$.
More specifically, for any $\vr z,\bar{\vr z}\in\prod_i\Omega_i$ of the form $\vr z = (z_1,z_2,\ldots,z_n),\bar{\vr z} = (\bar z_1,\bar z_2,\ldots,\bar z_n)$ we have
%\textcolor{red}{(revise the notation)}
\begin{equation*}
\mathbb P(\bar{\vr z}|\vr z) = \prod_iT_i(\bar X_i = \bar z_i|Pa(\bar X_i) = \vr z).
\end{equation*}

%\alex{[Here we don't really discuss any PMC procedure]}
%We construct a dynamic Bayesian network (DBN) from the obtained Bayesian network $\mathfrak B$. 
%%
%Assume that the states in the top layer of $\mathfrak B$ represent the states of the DBN at time $t$ and the states in the bottom layer show the states of the DBN at time $(t+1)$. 
%Unfold $\mathfrak B$ to obtain a DBN with exactly $N+1$ layers. \alex{[This unfolding operation is very clear, but not formally introduced. If we introduce the DBN notion formally (as hinted at before, earlier?), this could be easily fixed.]} 
%\textcolor{red}{Discuss graph for the initial states.} 

Our algorithm for probabilistic invariance computes $p_N(\vr z_0, Z_{\mathfrak a})$ to approximate $p_N(\vr s_0, A)$, for suitable choices of $\vr z_0$ and $Z_{\mathfrak a}$
depending on $\vr s_0$ and $A$.
% The initial state $\vr z(0) = \vr z_0$ of the DBN and the safe set $Z_{\mathfrak a}$ are selected with respect to the initial state $\vr s(0) = \vr s_0$ of $\mathscr M_{\mathfrak s}$ and the safe set $A$.  
The natural choice for the initial state is $\vr z_0 =(z_1(0),\ldots,z_n(0))$ with $z_i(0) =  \xi_i(\Pi_i(\vr s_0))$.
%This DBN is model checked against a safety specification over the time horizon $N$ and the safe set $Z_{\mathfrak a}$ constructed as follows.
%The safe set $Z_{\mathfrak a}$ is constructed as follows.
For $A$, the $n$-fold Cartesian product of the collection of the partition sets $\{D_{ij}\},\,i\in\mathbb N_n$ generates a cover of $A$ as 
\begin{align*}
	A& \subset \bigcup\{D_{1j}\}_{j=1}^{n_1}\times \{D_{2j}\}_{j=1}^{n_2}\times\ldots \times\{D_{nj}\}_{j=1}^{n_n}\\
	& = \bigcup_{\vr j}\left\{D_{\vr j}| \vr j = (j_1,j_2,\ldots,j_n), D_{\vr j} \doteq D_{1j_1}\times D_{2j_2}\times\ldots\times D_{nj_n}\right\}.
\end{align*}
We define the safe set $Z_{\mathfrak a}$ of the DBN as 
%\alex{[this does not seem to be a set over the product space where the DBN lives \ldots]}
\begin{equation}
\label{eq:safe_set_dis}
	Z_{\mathfrak a} = \bigcup_{\vr j}\left\{(z_{1j_1},z_{2j_2},\ldots,z_{nj_n}),\text{ such that } A\cap D_{\vr j}\ne\emptyset\text{ for }\vr j = (j_1,j_2,\ldots,j_n)\right\},
\end{equation}
which is a discrete representation of the continuous set $\bar A\subset \mathbb R^n$
\begin{equation}
	\label{eq:mod_safe_set}
	\bar A = \bigcup_{\vr j}\left\{D_{\vr j},\text{ such that } \vr j = (j_1,j_2,\ldots,j_n), A\cap D_{\vr j}\ne\emptyset\right\}.
\end{equation}
For instance $\bar A$ can be a finite union of hypercubes in $\mathbb R^n$ if the partition sets $D_{ij}$ are intervals.
It is clear that the set $\bar A$ is in general different form $A$. 

There are thus two sources of error: first due to replacing $A$ with $\bar A$, and second, due to the abstraction of the dynamics
% 
% This difference between $A$ and $\bar A$ adds to the error due to the abstraction of the dynamics of the concrete model $\mathscr M_{\mathfrak s}$: 
% both lead to the difference between quantity $p_N(\vr z_0,Z_{\mathfrak a})$ and $p_N(\vr s_0, A)$,   
% that is 
% %\alex{[We're mixing here the abstraction error with the MC error -- this relates to an earlier point and will confuse the readers.]}
% %
between the discrete outcome obtained by Theorem \ref{thm:safety_dis} and the continuous solution that results from \eqref{eq:bellman_rec}. 
In the next section we provide a quantitative bound on the two sources of error.

\subsection{Quantification of the Error due to Abstraction}
\label{subsec:error}

% Recall that our abstraction approach is based on partitioning each dimension of the state space separately. 
% % In order to take all the approximation steps formally we are required to preserve the labelling structure of the state space.
% % Therefore, it is essential that the safe set is representable by a finite union of sets in the $n$-dimensional Euclidean space $\mathbb R^n$ that are Cartesian product of exactly $n$ sets.
% If the safe set can not be represented in this form, we have to replace it with such a set $\bar A$.
% As discussed before, $A$ is replaced by $\bar A$ in \eqref{eq:mod_safe_set}, 
% which is of the form
% \begin{equation}
% \label{eq:cart_prod}
% \bar A = \bigcup_{i\in\Gamma} A_1^i\times A_2^i\times\ldots \times A_n^i,\quad A_j^i\subset\mathbb R,j\in\mathbb N_n,i\in\Gamma,
% \end{equation}
% where $\Gamma$ is a finite index set. 
% For instance $\bar A$ can be a finite union of hypercubes in $\mathbb R^n$ if the partition sets $D_{ij}$ are intervals.
% Clearly this approximation results in an error, which we can fully characterise in  
% %We first characterize this error in 
% Theorem \ref{thm:err_set}.

Let us explicitly write the Bellman recursion \eqref{eq:bellman_rec} of the safety problem over the set $\bar A$:
\begin{equation}
\label{eq:bellman_2}
W_N(\vr s) = \vr 1_{\bar A}(\vr s),\quad
W_k(\vr s) = \int_{\bar A}W_{k+1}(\vr{\bar s})t_{\mathfrak s}(\vr{\bar s}|\vr s)d\vr{\bar s},\quad
k=0,1,2,\ldots,N-1,
\end{equation}
which results in $p_N(\vr s_0,\bar A) = W_0(\vr s_0)$. 
Theorem~\ref{thm:err_set} characterizes the error due to replacing the safe set $A$ by $\bar A$.

\begin{theorem}
	\label{thm:err_set}
	Solution of the probabilistic invariance problem with the time horizon $N$ and two safe sets $A,\bar A$ satisfies the inequality
	\begin{equation*}
	|p_N(\vr s_0,A)-p_N(\vr s_0,\bar A)|\le MN\mathcal L(A\Delta\bar A),
	%\mathcal L(\bar A\backslash A),
	\quad \forall \vr s_0\in A\cap\bar A,
	%\quad \forall \vr s_0\in A,
	\end{equation*}
	where
	$M \doteq \sup\left\{t_{\mathfrak s}(\vr{\bar s}|\vr s)\big|\vr s,\vr{\bar s}\in A\Delta\bar A\right\}$.
	%$M = \sup\left\{t_{\mathfrak s}(\vr{\bar s}|\vr s)\big|\vr s,\vr{\bar s}%\in A\Delta\bar A\right\}$.
	%\in \bar A\backslash A\right\}$
	$\mathcal L(B)$ denotes the Lebesgue measure of any set $B\in\mathcal B$
	and
	$A\Delta\bar A \doteq (A\backslash\bar A)\cup(\bar A\backslash A)$
	is the symmetric difference of the two sets $A,\bar A$.
\end{theorem}

The second contribution to the error is related to the discretization of Algorithm \ref{algo:DBN_app} which is quantified
%in Theorem \ref{thm:err_dis}
by posing regularity conditions on the dynamics of the process.
The following Lipschitz continuity assumption restricts the generality of the density functions $t_k$ characterizing
the dynamics of model $\mathscr M_{\mathfrak s}$.
\begin{assumption}
	\label{ass:Lip_cont}
	Assume the density functions $t_k(\bar s_i|\cdot)$ are Lipschitz continuous with the finite positive $d_{ij}$
	\begin{equation*}
	|t_j(\bar s_j|\vr s)-t_j(\bar s_j|\vr{s'})|\le d_{ij}| s_i-s_i'|,
	\end{equation*}
	with $\vr s = [s_1,\ldots,s_{i-1},s_i,s_{i+1},\ldots,s_n]$ and $\vr{s'} = [s_1,\ldots,s_{i-1},s_i',s_{i+1},\ldots,s_n]$,
	for all $s_k,s_k',\bar s_k\in D_k$, $k\in\mathbb N_n$, and for all $i,j\in\mathbb N_n$.
\end{assumption}
%Assumption \ref{ass:Lip_cont} allows us to derive the following bound on the abstraction error.
Note that Assumption \ref{ass:Lip_cont} holds with $d_{ij} = 0$ if and only if $(X_i,\bar X_j)\notin\mathcal E$ in the DAG of the BN $\mathfrak B_{\rightarrow}$.
Assumption \ref{ass:Lip_cont} enables us to assign non-zero weights to the arcs of the graph and turn it into a weighted DAG.
The non-zero weight $w_{ij} = d_{ij}\mathcal L(D_j)$ is assigned to the arc $(X_i,\bar X_j)\in\mathcal E$, for all $i,j\in\mathbb N_n$.
We define the out-weight of the node $X_i$ by $\mathcal O_i = \sum_{j=1}^{n}w_{ij}$ and the in-weight of the node $\bar X_j$ by $\mathcal I_j = \sum_{i=1}^{n}w_{ij}$.
\begin{remark}
	The above assumption implies Lipschitz continuity of the conditional density functions $t_j(\bar s_j|\vr s)$.
	Since trivially $|s_i - s'_i| \le \|\vr s - \vr s'\|$ for all $i \in \mathbb N_n$, we obtain 
	\begin{equation*}
	|t_j(\bar s_j|\vr s)-t_j(\bar s_j|\vr s')|\le \mathcal H_j \|\vr s-\vr s'\|\quad \forall \vr s,\vr s'\in \bar A,\bar s_j\in D_j,
	\end{equation*}
	where $\mathcal H_j = \sum_{i=1}^{n}d_{ij}$.
	The density function $t_{\mathfrak s}(\vr{\bar s}|\vr s)$ is also Lipschitz continuous if the density functions $t_j(\bar s_j|\vr s)$ are bounded,
	but the boundedness assumption is not necessary for our result to hold.
\end{remark}

Assumption \ref{ass:Lip_cont} enables us to establish Lipschitz continuity of the value functions $W_k$ in \eqref{eq:bellman_2}.
%used in Bellman recursion \eqref{eq:bellman_rec} of the safety problem over the set $\bar A$ \myblue{(see \eqref{eq:bellman_2} in the Appendix for the explicit definition of $W_k$)}.
This continuity property is essential in proving an upper bound on the discretization error of Algorithm \ref{algo:DBN_app}, which is presented in Corollary \ref{cor:global_error}.
\begin{lemma}
	\label{lem:cont_val}
	Consider the value functions $W_k(\cdot)$, $k=0,1,2,\ldots,N$, employed in Bellman recursion \eqref{eq:bellman_2} of the safety problem over the set $\bar A$.
	Under Assumption \ref{ass:Lip_cont}, these value functions are Lipschitz continuous
	\begin{equation*}
	|W_k(\vr s)-W_k(\vr s')|\le \kappa \|\vr s-\vr s'\|,\quad\forall\vr s,\vr s'\in \bar A,
	\end{equation*}
	for all $k=0,1,2,\ldots,N$ with the constant $\kappa = \sum_{j=1}^{n}\mathcal I_j$,
	where $\mathcal I_j$ is the in-weight of the node $\bar X_j$ in the DAG of the BN $\mathfrak B_{\rightarrow}$.
\end{lemma}

\begin{corollary}
	\label{cor:global_error}
	The following inequality holds under Assumption \ref{ass:Lip_cont}:
	\begin{equation*}
		|p_N(\vr s_0, A)-p_N(\vr z_0, Z_{\mathfrak a})|\le MN\mathcal L(A\Delta\bar A) + N\kappa\delta \quad\forall \vr s_0\in A,
	\end{equation*}
	where $p_N(\vr z_0, Z_{\mathfrak a})$ is the invariance probability for the DBN obtained by Algorithm \ref{algo:DBN_app}. The initial state of the DBN is $\vr z_0 =(z_1(0),\ldots,z_n(0))$ with $z_i(0) =  \xi_i(\Pi_i(\vr s_0))$. The set $Z_{\mathfrak a}$ and the constant $M$ are defined in \eqref{eq:safe_set_dis} and Theorem \ref{thm:err_set}, respectively. The diameter of the partition of Algorithm \ref{algo:DBN_app} is defined and used as
	$\delta = \sup\{\|\vr s-\vr s'\|,\forall \vr s,\vr s'\in  D_{\vr j},\forall \vr j\,\,\,D_{\vr j}\subset \bar A\}.$
\end{corollary}

The second error term in Corollary \ref{cor:global_error} is a linear function of the partition diameter $\delta$, which depends on all partition sets along different dimensions. We are interested in proving a dimension-dependent error bound in order to parallelize the whole abstraction procedure along different dimensions. The next theorem gives this dimension-dependent error bound.
\begin{theorem}
	\label{thm:err_dis}
	The following inequality holds under Assumption \ref{ass:Lip_cont}:
	\begin{equation}
	|p_N(\vr s_0, A)-p_N(\vr z_0, Z_{\mathfrak a})|\le MN\mathcal L(A\Delta\bar A) + N\sum_{i=1}^{n} \mathcal O_i\delta_i \quad\forall \vr s_0\in A,
	\end{equation}
	with the constants defined in Corollary \ref{cor:global_error}.
	$\mathcal O_j$ is the out-weight of the node $X_i$ in the DAG of the BN $\mathfrak B_{\rightarrow}$.
	The quantity $\delta_i$ is the maximum diameter of the partition sets along the $i^{th}$ dimension
	$
	\delta_i = \sup\{|s_i-s_i'|,\forall s_i,s_i'\in  D_{ij},\forall j\in\mathbb N_{n_i}\}.
	$
\end{theorem}
For a given error threshold $\epsilon$, we can select the set $\bar A$ and consequently the diameters $\delta_i$ such that
$MN\mathcal L(A\Delta\bar A) + N\sum_{i=1}^{n} \mathcal O_i\delta_i\le \epsilon$.
Therefore, generation of the abstract DBN, namely selection of the partition sets $\{D_{ij},\,j\in\mathbb N_i\}$ (according to the diameter $\delta_i$) and computation of the CPD, can be implemented in parallel. 
For a given $\epsilon$ and set $\bar A$, the cardinality of the state space $\Omega_i, i\in\mathbb N_n,$ of the discrete random variable $X_i$ and thus the size of the CPD $T_i$, 
grow linearly as a function of the horizon of the specification $N.$ 

\section{Efficient Model Checking of the Finite-State DBN}
\label{sec:eff_comp}

Existing numerical methods for model checking DBNs with discrete r.v. transform the DBN into an explicit matrix representation \cite{JCLLZ09,L09,PT12},
which defeats the purpose of a compact representation. 
%, which is infeasible with a large number of random variables.
%
%In this section we show how factor graphs and the sum-product algorithm, developed for belief propagation in Bayesian networks, can be used
%to model check probabilistic invariance properties without explicitly constructing the transition matrix.
%The multiplicative structure of the transition probabilities in DBN can be incorporated in the computations, which results in reduction in memory usage.
Instead, we show that the multiplicative structure of the transition probability matrix can be incorporated in the computation which makes 
the construction of $\mathbb P(\bar{\vr z}|\vr z)$ dispensable. 
For this purpose we employ \emph{factor graphs} and the \emph{sum-product algorithm} \cite{KFL01} originally developed for marginalizing functions and applied to belief propagation in Bayesian networks. 
Suppose that a \emph{global} function is given as a product of \emph{local} functions, and that each local function depends on a subset of the variables of the global map.  
In its most general form, the sum-product algorithm acts on factor graphs in order to marginalize the global function, i.e., taking summation respect to a subset of variables, 
exploiting its product structure \cite{KFL01}. 
%We refer the reader to  for the formal definitions and apply them to the probabilistic invariance problem. 
In our problem, 
we restrict the summation domain of the Bellman recursion \eqref{eq:bellman_dis} to $\prod_i Z_i$ because the value functions are simply equal to zero in the complement of this set. 
The summand in \eqref{eq:bellman_dis} has the multiplicative structure 
\begin{equation}
\label{eq:fun_g}
g(\vr z,\bar{\vr z}) \doteq \vr 1_{Z_{\mathfrak a}}(\vr z)V_{k+1}^d(\bar{\vr z})\prod_i T_i(\bar X_i = \bar z_i|Pa(\bar X_i) = \vr z),
\quad V_k^d(\vr z) = \sum_{\bar{\vr z}\in \prod_i Z_i}g(\vr z,\bar{\vr z}).
\end{equation}
The function $g(\vr z,\bar{\vr z})$ depends on variables $\{z_i,\bar z_i,\,i\in\mathbb N_n\}$.
The factor graph of $g(\vr z,\bar{\vr z})$ has $2n$ \emph{variable nodes}, one for each variable and $(n+2)$ \emph{function nodes} for local functions $\vr 1_{Z_{\mathfrak a}},V_{k+1}^d,T_i$. An arc connects a variable node to a function node if and only if the variable is an argument of the local function.
The factor graph of Example \ref{exm:linear_system} for $n = 4$ is presented in Figure \ref{fig:factor_graph} -- 
factor graphs of general functions $g(\vr z,\bar{\vr z})$ in \eqref{eq:fun_g} are similar to that in Figure \ref{fig:factor_graph}, 
the only part needing to be modified being the set of arcs connecting variable nodes $\{z_i,\,i\in\mathbb N_n\}$ and function nodes $\{T_i,\,i\in\mathbb N_n\}$. 
This part of the graph can be obtained from the DAG of $\mathfrak B_{\rightarrow}$ of the DBN.

\begin{figure}
	\centering
	%\parbox{0.38\textwidth}{
	\begin{minipage}{0.38\textwidth}
		\includegraphics[width = \textwidth, height=0.2\textheight]{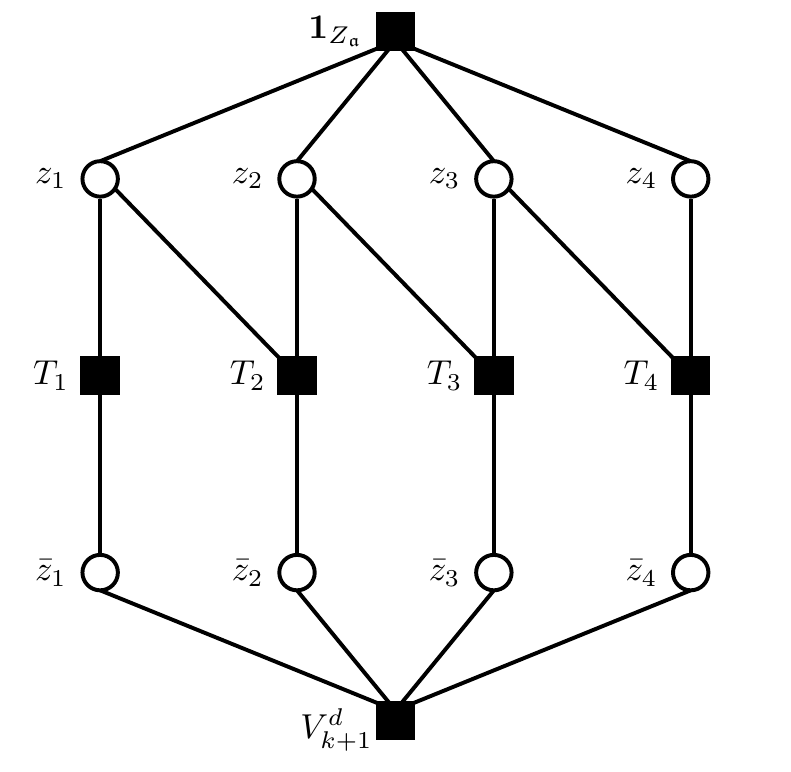}
		\caption{Factor graph of the linear stochastic system \eqref{eq:exmp_linear} for $n=4$.}
		\label{fig:factor_graph}
	\end{minipage}
	\quad
	\begin{minipage}{0.57\textwidth}
		\includegraphics[width = \textwidth]{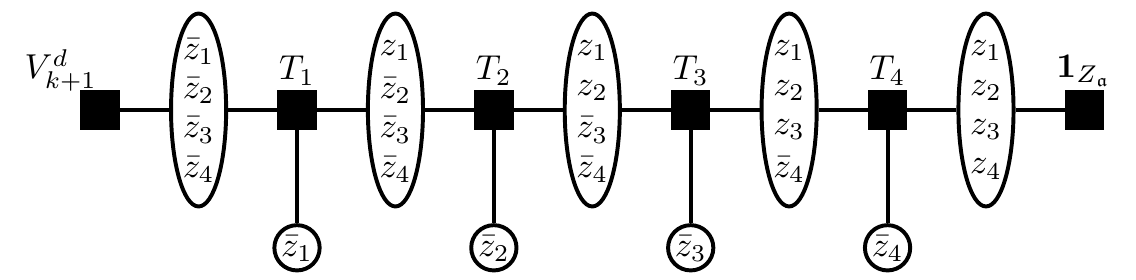}
		\includegraphics[width = 0.9\textwidth]{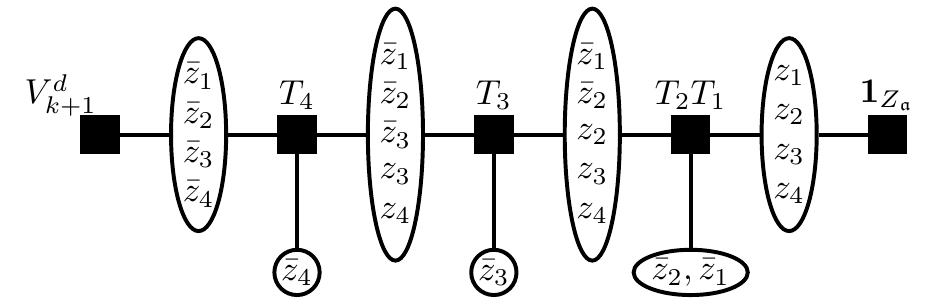}	
		\caption{Spanning tree of the linear stochastic system in \eqref{eq:exmp_linear} for $n=4$ and two orderings $(\bar z_4,\bar z_3,\bar z_2,\bar z_1)$ (top plot) and $(\bar z_1,\bar z_2,\bar z_3,\bar z_4)$ (bottom plot).}
		\label{fig:spanning_tree}
	\end{minipage}
\end{figure}

The factor graph of a function $g(\vr z,\bar{\vr z})$ contains loops for $n\ge 2$ and must be transformed to a \emph{spanning tree} using clustering and stretching transformations \cite{KFL01}. 
%These loops must be eliminated by stretching transformations of the variable nodes $z_i$.
For this purpose the order of clustering function nodes $\{T_i,i\in\mathbb N_n\}$ and that of stretching variable nodes $\{z_i,\, i\in\mathbb N_n\}$ needs to be chosen.
Figure \ref{fig:spanning_tree} presents the spanning trees of the stochastic system in \eqref{eq:exmp_linear} for two such orderings.
The variable nodes at the bottom of each spanning tree specify the order of the summation, 
whereas the function nodes considered from the left to the right indicate the order of multiplication of the local functions. 
The rest of the variable nodes show the arguments of the intermediate functions, 
which reflects the required memory for storing such functions. 
The computational complexity of the solution carried out on the spanning tree clearly depends on this ordering.

Algorithm \ref{algo:DBN_ordering} presents a greedy procedure that operates on the factor graph and provides an ordering of the variables and of the functions, in order to reduce the overall memory usage. 
This algorithm iteratively combines the function nodes and selects the next variable node, over which the summation is carried out.
The output of this algorithm implemented on the factor graph of Example \ref{exm:linear_system} is the orderings $\kappa_{\mathfrak f} = (\bar z_4,\bar z_3,\bar z_2,\bar z_1)$ and $e_{\mathfrak f} = (T_4,T_3,T_2,T_1)$, started from the outermost sum, which is related to the spanning tree on top of Figure \ref{fig:spanning_tree}.
\begin{algorithm}[t]
	\caption{Greedy algorithm for obtaining the order of stretching variables and clustering functions in the factor graph}
	\label{algo:DBN_ordering}
	\begin{center}
		\begin{algorithmic}[1]
			\REQUIRE 
			Factor graph of the summand in Bellman recursion
			\STATE
			Initialize the sets $\mathcal U_1 = \{z_i,\,i\in\mathbb N_n\}$,
			$\mathcal U_2 = \{\bar z_i,\,i\in\mathbb N_n\}$,
			$\mathcal U_3 = \{T_i,\,i\in\mathbb N_n\}$,
			$e_{\mathfrak f} = \kappa_{\mathfrak f} = \emptyset$
			\WHILE{$\mathcal U_1\neq\emptyset$}
			\STATE
			For any node $u\in\mathcal U_3$ compute $Pa_{\mathfrak f}(u)$ (resp. $Ch_{\mathfrak f}(u)$) as the elements of $\mathcal U_1$ (resp. $\mathcal U_2$) connected to $u$ by an arc in the factor graph
			\STATE
			Define the equivalence relation $R$ on $\mathcal U_3$ as $uR\bar u$ iff $Pa_{\mathfrak f}(u) = Pa_{\mathfrak f}(\bar u)$
			\STATE
			Replace the set $\mathcal U_3$ with the set of equivalence classes induced by $R$.
			\STATE
			Combine all the variable nodes of $Ch_{\mathfrak f}(u)$ connected to one class
			\STATE
			Select $u\in\mathcal U_3$ with the minimum cardinality of $Pa_{\mathfrak f}(u)$ and put $e_{\mathfrak f} = (u,e_{\mathfrak f}),\kappa_{\mathfrak f} = (Ch_{\mathfrak f}(u),\kappa_{\mathfrak f})$\STATE
			Update the sets $\mathcal U_1 = \mathcal U_1\backslash Pa_{\mathfrak f}(u)$,
			$\mathcal U_2 = \mathcal U_2\cup Pa_{\mathfrak f}(u)\backslash Ch_{\mathfrak f}(u)$,
			$\mathcal U_3 = \mathcal U_3\backslash\{u\}$, and eliminate all the arcs connected to $u$
			\ENDWHILE
			\ENSURE
			The order of variables $\kappa_{\mathfrak f}$ and functions $e_{\mathfrak f}$
		\end{algorithmic}
	\end{center}
\end{algorithm}

%\begin{figure}
%	\centering
%		\includegraphics[scale = 0.4]{stretched_graph_1.pdf}
%		\includegraphics[scale = 0.4]{stretched_graph_2.pdf}	
%	\caption{Spanning tree of the linear stochastic system in \eqref{eq:exmp_linear} for $n=4$ and two orderings $(\bar z_1,\bar z_2,\bar z_3,\bar z_4)$ (top plot) and $(\bar z_4,\bar z_3,\bar z_2,\bar z_1)$ (bottom plot).}
%	\label{fig:spanning_tree}
%\end{figure}
%\begin{figure}
%	\centering
%	\includegraphics[scale=0.8]{factor_graph.pdf}
%	\caption{Factor graph of the linear stochastic system \eqref{eq:exmp_linear} for $n=4$.}
%	\label{fig:factor_graph}
%\end{figure}

\section{Comparison with the State of the Art}
\label{sec:case_study}

In this section we compare our approach with the state-of-the-art abstraction procedure presented in \cite{APKL10} (referred to as $\mathrm{AKLP}$ in the following), 
which does not exploit the structure of the dynamics. 
The $\mathrm{AKLP}$ algorithm approximates the concrete model with a finite-state Markov chain by uniformly gridding the safe set.
As in our work, the error bound of the $\mathrm{AKLP}$ procedure depends on the global Lipschitz constant of the density function of the model, 
however it does not exploit its structure as proposed in this work. 
We compare the two procedures on (1) error bounds and (2) computational resources.

Consider the stochastic linear dynamical model in \eqref{eq:exmp_linear}, 
where $\Phi = [a_{ij}]_{i,j}$ is an arbitrary matrix.
The Lipschitz constants $d_{ij}$ in Assumption \ref{ass:Lip_cont} can be computed as 
%\begin{equation*}
%d_{ij} = \frac{|a_{ji}|}{\sigma_j^2\sqrt{2\pi e}},
%\end{equation*}
$d_{ij} = |a_{ji}|/\sigma_j^2\sqrt{2\pi e}$, where $e$ is Euler's constant.
From Theorem \ref{thm:err_dis}, we get the following error bound:
\begin{equation*}
	e_{\mathrm{DBN}} \doteq MN\mathcal L(A\Delta\bar A)+\frac{N}{\sqrt{2\pi e}}\sum_{i,j=1}^n\frac{|a_{ji}|}{\sigma_j^2}\mathcal L(D_j)\delta_i.
\end{equation*} 
On the other hand, the error bound for $\mathrm{AKLP}$ is
\begin{equation*}
	e_{\mathrm{AKLP}} = MN\mathcal L(A\Delta\bar A)
	+\frac{Ne^{-1/2}}{(\sqrt{2\pi})^n\sigma_1\sigma_2\ldots\sigma_n}\|\Sigma^{-1/2}\Phi\|_2\delta\mathcal L(A).
\end{equation*}
In order to meaningfully compare the two error bounds, 
select set $A = [-\alpha,\alpha]^n$ and $\sigma_i = \sigma,i\in\mathbb N_n$,
and consider hypercubes as partition sets. The two error terms then become
\begin{equation*}
	e_{\mathrm{DBN}} = \varsigma n\eta \left(\frac{\|\Phi\|_1}{n\sqrt{n}}\right),
	\quad e_{\mathrm{AKLP}} = \varsigma\eta^n\|\Phi\|_2,
	\quad \eta = \frac{2\alpha}{\sigma\sqrt{2\pi}},
	\quad \varsigma = \frac{N\delta}{\sigma\sqrt{e}},
\end{equation*}
where $\|\Phi\|_1$ and $\|\Phi\|_2$ are the entry-wise one-norm and the induced two-norm of matrix $\Phi$, respectively.
The error $e_{\mathrm{AKLP}}$ depends exponentially on the dimension $n$ as $\eta^n$, 
whereas we have reduced this term to a linear one $(n\eta)$ in our proposed new approach resulting in error $e_{\mathrm{DBN}}$. 
Note that $\eta\le 1$ means that the standard deviation of the process noise is larger than the selected safe set:  
in this case the value functions (which characterize the probabilistic invariance problem) uniformly converge to zero with rate $\eta^n$; 
clearly the case of $\eta>1$ is more interesting. 
% with respect to numerical algorithms. 
On the other hand for any matrix $\Phi$ we have $\frac{\|\Phi\|_1}{n\sqrt{n}}\le\|\Phi\|_2$. 
%\alex{[Thus far, all clear. Can you elaborate on this latter part of the discussion, and relate it precisely to the details in the Appendix? How do you exactly define the degree term $r$? I gather r is upper bounded by n -- this is not reflect in the proportionality relations below. ]}
This second term indicates how sparsity is reflected in the error computation. 
Denote by $r$ the degree of connectivity of the DAG of $\mathfrak B_{\rightarrow}$ for this linear system,
which is the maximum number of non-zero elements in rows of matrix $\Phi$.
We apply Lemma \ref{lem:norm_mat} in the Appendix to matrix $\Phi$ to get the inequalities
\begin{equation*}
	\|\Phi\|_2\le \sqrt{nr}\max_{i,j}|a_{ij}|, %\propto \sqrt{r},
	\qquad
	\frac{\|\Phi\|_1}{n\sqrt{n}}\le \frac{r}{\sqrt{n}}\max_{i,j}|a_{ij}|, 
	%\propto r,
\end{equation*}
which shows that for a fixed dimension $n$, sparse dynamics, compared to fully connected dynamics, 
results in better error bounds in the new approach. 

In order to compare computational resources, consider the numerical values $N = 10$, $\alpha = 1$, $\sigma = 0.2$, and the error threshold $\epsilon = 0.2$ for the lower bidiagonal matrix $\Phi$ with all the non-zero entries set to one. 
Table~\ref{tab:numerics} compares the number of required partition sets (or bins) per dimension, the number of marginals, and the required number of (addition and multiplication) operations for the verification step, for models of different dimensions (number of continuous variables $n$). 
The numerical values in Table~\ref{tab:numerics} confirm that for a given upper bound on the error $\epsilon$, the number of bins per dimension and the required marginals grow exponentially in dimension for $\mathrm{AKLP}$ and polynomially for our DBN-based approach.
For instance, to ensure the error is at most $\epsilon$ for the model of dimension $n=4$, 
the cardinality of the partition of each dimension for the uniform gridding and for the structured approach is $2.9\times 10^5$ and $8.5\times 10^3$, respectively. 
Then, $\mathrm{AKLP}$ requires storing $4.8\times 10^{43}$ entries (which is infeasible!), 
whereas the DBN approach requires $1.8\times 10^{12}$ entries ($\sim 8$GB). 
The number of
%(addition and multiplication)
operations required for computation of the safety probability are $1.1\times 10^{45}$ and $3.5\times 10^{21}$, respectively.
This shows a substantial reduction in memory usage and computational time effort: 
with given memory and computational resources, the DBN-based approach in compare with $\mathrm{AKLP}$ promises to handle systems with dimension that is at least twice as large. 

\begin{table}
\centering
\caption{Comparison of the $\mathrm{AKLP}$ and the DBN-based algorithms, over the stochastic linear dynamical model in \eqref{eq:exmp_linear}.
The number of partition sets (or bins) per dimension, the number of marginals, and the total required number of (addition and multiplication) operations for the verification step, 
are compared for models of different dimensions (number of continuous variables $n$).} 
\resizebox{\columnwidth}{!}{ %
\begin{tabular}{|@{\hskip1pt}l@{\hskip1pt}|@{\hskip1pt}c@{\hskip1pt}|@{\hskip1pt}c@{\hskip1pt}|@{\hskip1pt}c@{\hskip1pt}|@{\hskip1pt}c@{\hskip1pt}|@{\hskip1pt}c@{\hskip1pt}|@{\hskip1pt}c@{\hskip1pt}|@{\hskip1pt}c@{\hskip1pt}|@{\hskip1pt}c@{\hskip1pt}|@{\hskip1pt}c@{\hskip1pt}|}
\hline
\multicolumn{2}{|c|}{dimension $n$}
 & 1 & 2 & 3 & 4 & 5 & 6 & 7 & 8\\
\hline
\multirow{2}{*}{ \# bins/dim}
&$\mathrm{AKLP}$ & $1.2\times 10^3$ & $1.1\times 10^4$ & $6.0\times 10^4$ & $2.9\times 10^5$ & $1.3\times 10^6$ & $5.8\times 10^6$ & $2.5\times 10^7$ & $1.1\times 10^8$\\
& $\mathrm{DBN}$ & $1.2\times 10^3$ & $3.6\times 10^3$ & $6.0\times 10^3$ & $8.5\times 10^3$ & $1.1\times 10^4$ & $1.3\times 10^4$ & $1.6\times 10^4$ & $1.8\times 10^4$\\
\hline
\multirow{2}{*}{ \# marginals}
&$\mathrm{AKLP}$ & $1.5\times 10^6$ & $1.5\times 10^{16}$ & $4.8\times 10^{28}$ & $4.8\times 10^{43}$ & $1.5\times 10^{61}$ & $1.5\times 10^{81}$ & $4.3\times 10^{103}$ & $3.5\times 10^{128}$ \\
& $\mathrm{DBN}$ & $1.5\times 10^6$ & $4.8\times 10^{10}$ & $4.4\times 10^{11}$ & $1.8\times 10^{12}$ & $5.2\times 10^{12}$ & $1.2\times 10^{13}$ & $2.3\times 10^{13}$ & $4.2\times 10^{13}$\\
\hline
\multirow{2}{*}{ \# operations}
&$\mathrm{AKLP}$ & $2.9\times 10^7$ & $3.1\times 10^{17}$ & $1.0\times 10^{30}$ & $1.1\times 10^{45}$ & $3.7\times 10^{62}$ & $3.7\times 10^{82}$ & $1.1\times 10^{105}$ & $9.5\times 10^{129}$ \\
& $\mathrm{DBN}$ & $2.9\times 10^7$ & $1.9\times 10^{12}$ & $8.0\times 10^{16}$ & $3.5\times 10^{21}$ & $1.7\times 10^{26}$ & $8.9\times 10^{30}$ & $5.2\times 10^{35}$ & $3.4\times 10^{40}$\\
\hline
\end{tabular}
}
\label{tab:numerics}
\end{table}

\section{Conclusions and Future Directions}
\label{sec:concl}

While we have focused on probabilistic invariance, our abstraction approach
can be extended to more general properties expressed within the bounded-horizon fragment of PCTL \cite{rcsl2010} or 
to bounded-horizon linear temporal properties \cite{TA13,tmka2013}, since the model checking problem for these
logics reduce to computations of value functions similar to the Bellman recursion scheme.
% While both extensions can be reduced to computations of value functions related those \eqref{eq:bellman_rec} for the probabilistic invariance problem, 
% it is necessary to adapt the study to the corresponding labelling functions. 
%
Our focus in this paper has been the foundations of DBN-based abstraction for general Markov processes: factored representations, error bounds, and algorithms.
We are currently implementing these algorithms in the \software tool \cite{FAUST15},
and scaling the algorithms using dimension-dependent adaptive gridding \cite{SA13} as well as implementations of the sum-product
algorithm on top of data structures such as algebraic decision diagrams (as in probabilistic model checkers \cite{KNP11}).

% In view of further mitigations to the state-space explosion problem, 
% the discussed dimension-dependent error formulation can be further ameliorated and employed to construct dimension-dependent adaptive gridding \cite{SA11,SA13}. 
% Finally, the computations of the sum-product algorithm can be implemented based on algebraic decision diagrams or sparse matrix multiplications. 
% We plan to include the discussed work and these extension in the software tool \software \cite{FAUST15}. 

%\subparagraph*{Acknowledgements}
%
%I want to thank \dots

\bibliographystyle{plain}
\bibliography{biblio}

\appendix
\section{Proof of Statements}
\label{sec:proofs}

\begin{proof}[Proof of Theorem \ref{thm:err_set}]
	%The particular selection of $\bar A$ implies $A\subset\bar A$.
	Recall the recursive equations for the probabilistic safety problem over sets $A$ and $\bar A$ as in \eqref{eq:bellman_rec} and \eqref{eq:bellman_2}, respectively.
%	are as follows
%	\begin{align*}
%		& V_N(\vr s) = \vr 1_A(\vr s),\quad
%		V_k(\vr s) = \int_A V_{k+1}(\vr{\bar s})t_{\mathfrak s}(\vr{\bar s}|\vr s)d\vr{\bar s},\nonumber\\
%		& W_N(\vr s) = \vr 1_{\bar A}(\vr s),\quad
%		W_k(\vr s) = \int_{\bar A}W_{k+1}(\vr{\bar s})t_{\mathfrak s}(\vr{\bar s}|\vr s)d\vr{\bar s},
%	\end{align*}
%	for all $\vr s\in\mathbb R^n, k=0,1,2,\ldots,N-1,$
%	where
	The solutions of the safety problems are $p_N(\vr s_0, A) = V_0(\vr s_0)$ and $p_N(\vr s_0,\bar A) = W_0(\vr s_0)$.
	We prove inductively that the inequality
	$|V_{k}(\vr s)- W_{k}(\vr s)|\le M(N-k)\mathcal L(\bar A\Delta A)$
	holds for all $\vr s\in A\cap\bar A$. This inequality is true for $k = N$.
	For any $k=0,1,2,\ldots,N-1$ and any state $\vr s\in A\cap\bar A$ we have
%	\begin{align*}
%		|V_{N-1}(\vr s) - W_{N-1}(\vr s)|\le \left|\int_A t_{\mathfrak s}(\vr{\bar s}|\vr s)d\vr{\bar s}-
%		\int_{\bar A}t_{\mathfrak s}(\vr{\bar s}|\vr s)d\vr{\bar s}\right|
%		%& =
%		%\left|\int_{A\backslash \bar A}t_{\mathfrak s}(\vr{\bar s}|\vr s)d\vr{\bar s}-
%		%\int_{\bar A\backslash A}t_{\mathfrak s}(\vr{\bar s}|\vr s)d\vr{\bar s}\right|
%		= \int_{\bar A\backslash A}t_{\mathfrak s}(\vr{\bar s}|\vr s)d\vr{\bar s}
%		\le M\mathcal L(\bar A\backslash A).
%	\end{align*}
	%For any $k=0,1,2,\ldots,N-2$, we have
	\begin{align*}
		|V_{k}(\vr s)- W_{k}(\vr s)|& \le \int_{A\cap\bar A}|V_{k+1}(\vr{\bar s})- W_{k+1}(\vr{\bar s})|t_{\mathfrak s}(\vr{\bar s}|\vr s)d\vr{\bar s}\\
		& + \int_{A\backslash \bar A}V_{k+1}(\vr{\bar s})t_{\mathfrak s}(\vr{\bar s}|\vr s)d\vr{\bar s}
		+ \int_{\bar A\backslash A}W_{k+1}(\vr{\bar s})t_{\mathfrak s}(\vr{\bar s}|\vr s)d\vr{\bar s}\\
		& \le M(N-k-1)\mathcal L(\bar A\Delta A) + M\mathcal L(\bar A\backslash A)+M\mathcal L(A\backslash\bar A)\\
		& = M(N-k)\mathcal L(\bar A\Delta A).
	\end{align*}
	The inequality for $k = 0$ proves upper bound $MN\mathcal L(\bar A\Delta A)$ on $|p_N(\vr s_0, A)-p_N(\vr s_0,\bar A)|$.
\end{proof}

\begin{proof}[Proof of Lemma \ref{lem:cont_val}]
The inequality holds for $k=N$. For $k=0,1,2,\ldots,N-1$ and any $\vr s,\vr s'\in\bar A$ we have
\begin{equation*}
|W_k(\vr s)-W_k(\vr s')|
\le\int_{\bar A} W_{k+1}(\bar{\vr s})|t_{\mathfrak s}(\bar{\vr s}|\vr s)-t_{\mathfrak s}(\bar{\vr s}|\vr s')|d\bar{\vr s}
\le \int_{\bar A}|t_{\mathfrak s}(\bar{\vr s}|\vr s)-t_{\mathfrak s}(\bar{\vr s}|\vr s')|d\bar{\vr s}
\end{equation*}
Next, we employ a telescopic sum for the multiplicative structure of the density functions in the integrand on the right-hand side, to obtain: 
\begin{align*}
|W_k(\vr s)-W_k(\vr s')|&\le \int_{\bar A}\left|\prod_{i=1}^{n}t_i(\bar s_i|\vr s)-\prod_{i=1}^{n}t_i(\bar s_i|\vr s')\right|d\bar{\vr s}\\
	& = \int_{\bar A}\left|\sum_{j=1}^{n}\left[\prod_{i=1}^{j-1}t_i(\bar s_i|\vr s')\prod_{i=j}^{n}t_i(\bar s_i|\vr s)
	-\prod_{i=1}^{j}t_i(\bar s_i|\vr s')\prod_{i=j+1}^{n}t_i(\bar s_i|\vr s)\right]\right|d\bar{\vr s}\\
	& \le \sum_{j=1}^{n}\int_{\bar A}\left[\prod_{i=1}^{j-1}t_i(\bar s_i|\vr s')\prod_{i=j+1}^{n}t_i(\bar s_i|\vr s)
	\left| t_j(\bar s_j|\vr s)-t_j(\bar s_j|\vr s')\right|\right]d\bar{\vr s}\\
	& \le \sum_{j=1}^{n}\int_{D_j}\left| t_j(\bar s_j|\vr s)-t_j(\bar s_j|\vr s')\right|d\bar s_j\\
	& \le \sum_{j=1}^{n}\mathcal H_j\|\vr s-\vr s'\|\mathcal L(D_j)
	=\|\vr s-\vr s'\|\sum_{j=1}^{n}\mathcal H_j\mathcal L(D_j) = \|\vr s-\vr s'\| \sum_{j=1}^{n}\mathcal I_j.
\end{align*}
\end{proof}

\begin{lemma}
	\label{lem:norm_mat}
	The entry-wise one-norm and two-norm of square matrices are equivalent:
	\begin{equation*}
		n\|\Phi\|_2\le \|\Phi\|_1\le n\sqrt{n}\|\Phi\|_2,
	\end{equation*}
	where $n$ is the dimension of the matrix $\Phi = [a_{ij}]_{i,j}\in\mathbb R^{n\times n}$.
\end{lemma}

\begin{proof}[Proof of Lemma \ref{lem:norm_mat}]
	Define $r_i(\Phi) = \sum_{j=1}^{n}|a_{ij}|$ and $c_j(\Phi) = \sum_{i=1}^{n}|a_{ij}|$. The Cauchy-Schwartz inequality implies that
	\begin{align*}
		& c_j(\Phi) \le \sqrt{n}\sqrt{\sum_{i=1}^n |a_{ij}|^2} = \sqrt{n}\|\Phi\vr e_1\|_2\le \sqrt{n}\|\Phi\|_2\\
		& \Rightarrow \|\Phi\|_1 = \sum_{j=1}^{n}c_j(\Phi)\le \sum_{j=1}^{n}\sqrt{n}\|\Phi\|_2 = n\sqrt{n}\|\Phi\|_2,
	\end{align*}
	where $\vr e_1 = [1,0,0,\ldots,0]^T$. On the other hand for any $\vr s = [s_1,s_2,\ldots,s_n]^T$ with $\|\vr s\|_2 = 1$,
	\begin{align*}
		\|\Phi\vr s\|_2 & = \left[\sum_{i=1}^{n}\left|a_{i1}s_1+a_{i2}s_2+\ldots+a_{in}s_n\right|^2\right]^{1/2}
		\le \left[\sum_{i=1}^{n}\left(|a_{i1}|^2+|a_{i2}|^2+\ldots+|a_{in}|^2\right)\right]^{1/2}\\
		&  = \left[\sum_{i,j=1}^{n}|a_{ij}|^2\right]^{1/2}\le \frac{n}{n^2}\sum_{i,j=1}^{n}|a_{ij}| = \frac{1}{n}\|\Phi\|_1.
	\end{align*}
\end{proof}

As you see here the ratio $\|\Phi\|_1/\|\Phi\|_2$ is bounded from below by the dimension of the matrix and also from above by the $n\sqrt{n}$.

\begin{lemma}[\cite{Kol06}]
	The maximum singular value of a matrix can be bounded based on its sparsity pattern. In particular for any $\Phi$,
	\begin{equation*}
		\|\Phi\|_2\le \max_{i,j:a_{ij}\ne 0}[r_i(\Phi)c_j(\Phi)]^{1/2}.
	\end{equation*}
\end{lemma}

\end{document}